\documentclass[final,11pt]{article}
\usepackage{amsmath}
\usepackage{amsfonts}
\usepackage{amsthm}

\newtheorem{theorem}{Theorem}[section]
\newtheorem{lemma}[theorem]{Lemma}

\begin{document}

A. G. Ramm, Inverse scattering on the half-line revisited,
Reports on Math. Physics (ROMP), 76, N2, (2015), 159-169. 

\title{Inverse scattering on the half-line revisited}

\author{A. G. Ramm\\
\small Mathematics Department\\
\small Kansas State University, Manhattan, KS 66506, USA\\
\small email: ramm@math.ksu.edu}
\date{}
\maketitle

\begin{abstract}

The inverse scattering problem on the half-line has been studied in the literature in detail. V. Marchenko presented the solution to this problem. In this paper, the invertibility of the steps of the inversion procedure is discussed and a new set of necessary and sufficient conditions on the scattering data is given for the scattering data to be generated by a potential $q \in L_{1,1}$. Our proof is new and in contrast with Marchenko's proof does not use equations on the negative half-line.

\end{abstract}



\section{Introduction} \label{Introduction}
Let $lu := -u'' + q(x)u$, where $q$ is a real-valued function such that $$\displaystyle\int_0^{\infty}x|q(x)|dx < \infty.$$
The class of such $q$ is called $L_{1,1}$. The domain $D(l) = \{ u: u \in H^2 = H^2(\mathbb{R}_{+}), u(0) = 0, lu \in L^2 = L^2(\mathbb{R}_{+}) \}$,
$\mathbb{R}_{+} := [0, \infty)$. The closure of the operator $l$ is self-adjoint, its continuous spectrum is $[0, \infty)$ and its
discrete spectrum consists of finitely many negative eigenvalues of finite multiplicity.

Let $\varphi := \varphi(x,k)$ be the solution to the problem $l\varphi = k^2\varphi$ in $\mathbb{R}_{+}$, $\varphi(0,k) = 0, \varphi'(0,k) = 1$,
and $f := f(x,k)$ solve the problem $lf = k^2f$ in $\mathbb{R}_{+}$, $f = e^{ikx} + o(1)$ as $x \to \infty$. It is known that $\varphi(0,k)$ is an entire
function of $k$ and $f(x,k)$ is an analytic function of $k$ in the half-plane Im$k>0$, (see \cite{L}, \cite{M} or \cite{R470}).

Let $\varphi_j(x) := \varphi(x,ik_j), f_j := f(x,ik_j)$, where $k_j > 0$ and $ik_j$ are simple zeros of $f(k) := f(0,k)$ in
 $C_{+} := \{ z: \text{Im}z > 0 \}$. The numbers $-k_j^2$, $1 \leq j \leq J$, are negative eigenvalues of the operator $l$.

It is known that
$$||\varphi_j||^2 := ||\varphi_j||^2_{L^2(\mathbb{R}_{+})} := c_j^{-1} = s_j^{-1}\left[ f'(0,ik_j)\right]^{-2},$$ where
$\varphi_j(x) = \displaystyle\frac{f_j(x)}{f'(0,ik_j)}$, $s_j^{-1} := ||f_j||^2$, $s_j>0$.

The eigenfunction expansion theorem is known (see \cite{L}, \cite{M} or [5]-[8]):
\begin{equation}\label{eq1}
\int_0^{\infty}\varphi(x,k)\varphi(y,k)\frac{2}{\pi}\frac{k^2dk}{|f(k)|^2} + \sum_{j = 1}^Js_jf_j(x)f_j(y) = \delta(x - y).
\end{equation}

Define the scattering data $\mathcal{S}$ to be the collection
\begin{equation}\label{eq2}
\mathcal{S} := \{ S(k), k_j, s_j, 1 \leq j \leq J \},
\end{equation}
where $S(k) := \frac{f(-k)}{f(k)}$, and
\begin{equation}\label{eq2'}
f(x,k) = e^{ikx} + \int_x^{\infty}A(x,y)e^{iky}dy = e^{ikx}\left( 1 + \int_0^{\infty}A(x,x+p)e^{ikp}dp \right).
\end{equation}

The {\em inverse scattering problem} (ISP) consists of finding $q(x)$ from the knowledge of $\mathcal{S}$.
 The ISP has been solved in \cite{M}, [5],[6]. This solution consists of the following three steps:
\begin{equation}\label{eq3}
\mathcal{S} \overset{1}{\Longrightarrow} F \overset{2}{\Longrightarrow} A \overset{3}{\Longrightarrow} q,
\end{equation}
where (this is step 1):
\begin{equation}\label{eq4}
F(x) := F_s + F_d := \frac{1}{2\pi}\int_{-\infty}^{\infty}[1 - S(k)]e^{ikx}dx + \sum_{j = 1}^Js_je^{-k_jx},
\end{equation}
$A := A(x,y)$ is the kernel to be found from the basic equation, also called Marchenko's equation, (this is step 2):
\begin{equation}\label{eq5}
A(x,y) + F(x + y) + \int_x^{\infty}A(x,s)F(s + y)ds = 0, \quad y \geq x,
\end{equation}
and the potential is found by the formula (this is step 3):
\begin{equation}\label{eq6}
q(x) = -2\dot{A} := -2\,\frac{dA(x,x)}{dx}.
\end{equation}
Note  that $q$ and $A(x,y)$ are related by the equation
\begin{equation}\label{eq6'}
A(x,y) = \frac{1}{2}\int_{\frac{x+y}{2}}^{\infty}q(s)ds + \frac{1}{2}\int_x^\infty dsq(s)\int_{y - s + x}^{y + s - x} A(s,u)du,
\end{equation}
(see \cite{M}, p.175). It is proved in \cite{M} (and in \cite{R470}) that if $q \in L_{1,1}$ then equation \eqref{eq5} is
uniquely solvable for $A(x,y)$ for any $x \geq 0$, the operator
$$Fh := \displaystyle\int_0^\infty F(s + y)h(s)ds$$
is compact in $L_p(\mathbb{R}_{+})$, $p = 1,2,\infty$, and
the operator $(I + F)^{-1}$ is bounded in $L_p(\mathbb{R}_{+})$, $p = 1,2,\infty$.

In \cite{M} {\em the  invertibility of the steps in the inversion scheme \eqref{eq3} is not discussed}. This invertibility is one of the topics discussed in our paper (see also \cite{R470}).
The other topic is {\em  necessary and sufficient conditions on the scattering data $\mathcal{S}$ for
 these data to correspond to a potential $q \in L_{1,1}$.} In the book \cite{M} on p. 234 in Theorem 3.3.3 it is claimed that
 conditions I and II (see p. 218 of \cite{M}) are necessary and sufficient for $\mathcal{S}$ to be generated by a $q \in L_{1,1}$.

In our paper the necessary and sufficient conditions on $\mathcal{S}$ for $q \in L_{1,1}$ are different from these in \cite{M}. Our proofs
contain many new arguments based on the theory of Riemann problem.
{\em In contrast with the arguments in \cite{M}, we use neither equations
 on $(-\infty, 0]$  nor the equation with the operator $I + F_{s,0}^{+}$ in $L^2(\mathbb{R}_{+})$ (see \cite{M}, p. 228, Theorem 3.3.2).}

In condition II on p. 218 in \cite{M} there is a misprint (the term $\displaystyle\frac{1 - S(0)}{1}$ should be replaced
 by $\displaystyle\frac{1 - S(0)}{4}$). We use a different condition:
\begin{equation}\label{eq7}
\text{ind}S(k) = \text{ind}\frac{f(-k)}{f(k)} = -2\text{ind}f(k) = \begin{cases} -2J &\mbox{if } f(0) \neq 0, \\
-2J - 1 & \mbox{if } f(0) = 0. \end{cases}
\end{equation}

Our necessary and sufficient conditions on the scattering data for these data to be generated by a $q\in L_{1,1}$ are:
\begin{eqnarray}
S(-k) = \overline{S(k)} = S^{-1}(k), \quad k \in \mathbb{R}; \quad S(\infty) = 1, \label{eq8} \\
k_j > 0, \quad s_j > 0, \quad 1 \leq j \leq J, \label{eq9} \\
F_s(x) \in L^1(\mathbb{R}), \quad xF' \in L^1(\mathbb{R}_+), \label{eq10} \\
\text{ind}S(k) \text{ is a non-positive integer.} \label{eq11}
\end{eqnarray}

Here the overline stands for complex conjugate,
\begin{equation}\label{eq12}
\text{ind}S(k) := \frac{1}{2\pi}\Delta_{(-\infty, \infty)}\text{arg}S(k),
\end{equation}
and $\Delta_{(-\infty, \infty)}\text{arg}S(k)$ is the increment of the argument of $S(k)$ when $k$ runs from $-\infty$ to $\infty$ along the real axis.

Our results can be formulated in the following theorems:
\begin{theorem}\label{Thm1}
All steps in the recovery scheme \eqref{eq2} are reversible.
\end{theorem}

\begin{theorem}\label{Thm2}
If $q \in L_{1,1}$ then assumptions \eqref{eq8} - \eqref{eq11} are satisfied. Conversely, if assumptions \eqref{eq8} - \eqref{eq11} hold, then the scattering data $\mathcal{S}$, defined in \eqref{eq2}, corresponds to $q \in L_{1,1}$.
\end{theorem}

A characterization of the scattering data for $q \in L_{1,1}$ has been proposed in \cite{M}, where this characterization is different
 from the one given in Theorem 1.2. There is no characterization of the scattering data corresponding to $q\in L_{1,1}$ in other widely used books on inverse scattering, for example, in \cite{L}.

In Section 2 proofs are given.

\section{Proofs}
\begin{proof}[Proof of Theorem 1.1]
Let us recall the known estimates: if $q \in L_{1,1}$, then
\begin{eqnarray}
|A(x,y)| \leq c\int_{\frac{x+y}{2}}^\infty |q(t)|dt := cz\left( \frac{x+y}{2} \right),  \label{eq13} \\
\quad |A(y)| := |A(0,y)| \leq c\int_{y/2}^\infty |q(t)|dt, \label{eq13'}\\
\int_0^\infty z(x)dx = \int_0^\infty t|q(t)|dt < \infty. \label{eq14}
\end{eqnarray}
Let us prove the invertibility of the steps in the inversion procedure \eqref{eq3}.

{\em To prove $F(x) \Rightarrow \mathcal{S}$}, let us take $x \to -\infty$ in \eqref{eq4} and assume that $0 < k_1 < k_2 \dots < k_J$.
 Then $k_J$ and $s_J$ are determined as the main term of the asymptotic of $F(x)$ as $x\to -\infty$.
   Consider $F(x) - s_Je^{-k_J x}$ and let $x \to -\infty$. Then $k_{J - 1}$ and $s_{J - 1}$ are determined
as the main term of the asymptotic of $F(x) - s_Je^{-k_J x}$ as $x\to -\infty$.  Repeat this argument $J$ times and get $k_1, s_1, \dots, k_J, s_J$, that is, $F_d(x)$. Thus, $F_s(x) = F(x) - F_d(x)$ is found. Taking the inverse Fourier transform of $F_s$, one finds $1 - S(k)$:
$$1 - S(k) = \int_{-\infty}^\infty F_s(x)e^{-ikx}dx.$$
Therefore, $S(k)$ is found, and the scattering data $\mathcal{S}$ are uniquely recovered.\\

{\em To prove $A \Rightarrow F$}, one considers \eqref{eq5} as a Volterra equation for $F$ with the kernel $A(x,y)$. Using \eqref{eq13} it is not difficult to prove that this equation is uniquely solvable by iterations and $F$ is uniquely determined if $A$ is given.

{\em Here is an alternative proof:}

Given $A(x,y)$, one constructs
\begin{equation}\label{eq15}
f(x,k) = e^{ikx} + \int_x^\infty A(x,y)e^{iky}dy,
\end{equation}
and finds $f(k):=f(0,k)$:
\begin{equation}\label{eq16}
f(k) = 1 + \int_0^\infty A(y)e^{iky}dy.
\end{equation}

The zeros of $f(k)$ in $\mathbb{C}_+$ are the numbers $ik_j$, $k_j > 0$, $1 \leq j \leq J$, and so the function $S(k) = \frac{f(-k)}{f(k)}$, the numbers $k_j > 0$ and the number $J$ are found. The numbers $s_j > 0$ are found uniquely by the formulas:
\begin{equation}\label{eq17}
s_j = ||f_j(x)||^{-2}, \quad f_j(x) = e^{-k_jx} + \int_x^\infty A(x,y)e^{-k_jy}dy,
\end{equation}
see also formula \eqref{eq20} below.\\

{\em To prove $q \Rightarrow A(x,y)$}, one can use the known equation for $A(x,y)$ (see, for example, \cite{M}, p. 175):
\begin{equation}\label{eq18}
A(x,y) = \frac{1}{2}\int_{\frac{x+y}{2}}^\infty q(t)dt + \frac{1}{2}\int_x^\infty ds q(s) \int_{y-s+x}^{y+s-x}A(s,t)dt,
\end{equation}
which is uniquely solvable by iterations for $A(x,y)$ if $q \in L_{1,1}$ is known.\\

Theorem \ref{Thm1} is proved.
\end{proof}

\textbf{Remark 1.} From Theorem \ref{Thm1} it follows that $q$, obtained by the inversion scheme \eqref{eq3}
( see formula \eqref{eq6}), is identical with the original $q$ that generated the scattering data $\mathcal{S}$.
Indeed, both $q$ have the same scattering data  by Theorem \ref{Thm1}, so they both have the same $A(x,y)$.
Therefore, these two potentials are identical: they are both calculated by formula \eqref{eq6}.\\

\begin{proof}[Proof of Theorem \ref{Thm2}.]$\\$
{\em Necessity.} Assume that $q \in L_{1,1}$. It is known (see \cite{M} or \cite{R470}) that the solution $f(x,k)$ is defined uniquely by the equation
\begin{equation}\label{eq19}
f(x,k) = e^{ikx} + \int_x^\infty \frac{\sin(k(y - x))}{k}q(y)f(y,k)dy,
\end{equation}
which is of Volterra type because $q \in L_{1,1}$.
  The data $\mathcal{S}$ can be constructed  from  $f(x,k)$: the zeros of $f(0,k) := f(k)$ in $\mathbb{C}_+$ are simple zeros
   $ik_j$, $k_j > 0$, $1 \leq j \leq J$, and, possibly, $k = 0$ is a simple zero. The function $S(k) = \frac{f(-k)}{f(k)}$ at
    $k = 0$ does not have zero:
$$S(0) = \begin{cases} 1 &\mbox{if } f(0) \neq 0, \\
-1 & \mbox{if } f(0) = 0. \end{cases}$$
Here the L' Hospital's rule and the simplicity of zero $k = 0$ were used: if $f(0)=0$ then
$$ S(0)=\lim_{k\to 0}\frac{f(-k)}{f(k)}=-\frac{\dot{f}(0)}{\dot{f}(0)}=-1.$$
If $q$ is real-valued, then $\overline{f(k)} = f(-k)$, $k \in \mathbb{R}$, so $\overline{S(k)} = S(-k) = S^{-1}(k)$, $|S(k)| = 1$.
Since $f(\infty) = 1$ (see \eqref{eq2'}) one has $S(\infty) = 1$. Let us define the phase shift by the formula
$f(k) = |f(k)|e^{-i\delta(k)}$, so $S(k) = e^{2i\delta(k)}$ and $\delta(-k) = -\delta(k)$, $k \in \mathbb{R}$. Finally, let
\begin{equation}\label{eq20}
s_j = - \frac{2ik_j f'(0, ik_j)}{\dot{f}(ik_j)},
\end{equation}
see \cite{M} or \cite{R470}. Consequently, $q \in L_{1,1}$ uniquely determines the scattering data $\mathcal{S}$, defined in \eqref{eq2}.
If $\mathcal{S}$ is given, one calculates $F(x)$ by formula \eqref{eq4}. The basic equation \eqref{eq5} is derived from formula \eqref{eq1}.
This derivation is known (see, for example, \cite{R470} p. 139) and by this reason is omitted.\\

It is known that equation \eqref{eq5} has a unique solution $A(x,y)$ in $L^p(\mathbb{R}_+)$, $p = 1,2$, for any $x \geq 0$, and the operator
 $Fh := \int_x^\infty F(s + y)h(s)ds$ is compact in $L^p(\mathbb{R}_+), p = 0,1,2$, for any $x \geq 0$ (see \cite{M} or \cite{R470}).
  If $A(x,y)$ is found, then $q$ is found by formula \eqref{eq6}, and the recovery process of finding $q$ from the scattering data
  $\mathcal{S}$ is completed.
 We have checked above that conditions \eqref{eq8} and \eqref{eq9} hold. Condition \eqref{eq11} follows from formula \eqref{eq7}. Let us check
  conditions \eqref{eq10}. Write equation \eqref{eq5} as
\begin{equation}\label{eq21}
A(x, p-x) + F(p) + \int_p^\infty A(x, x + t - p)F(t)dt = 0.
\end{equation}
Here the substitutions $p = x+y$, $t = s+y$ were made. Consider the operator
\begin{equation}\label{eq22}
Th := \int_p^\infty A(x, x + t - p)h(t)dt
\end{equation}
as an operator in $L^1(2x, \infty)$ since $p = x+y \geq 2x$. For any $x \geq 0$ this operator is compact. Its norm is bounded
\begin{eqnarray}
||T|| := ||T||_{L^1(2x, \infty) \to L^1(2x, \infty)} \leq \int_0^\infty |A(x,x+u)|du \nonumber\\
\leq c\int_0^\infty du \int_{x + \frac{u}{2}}^\infty |q(s)|ds \leq 2c\int_x^\infty s |q(s)|ds.\label{eq23}
\end{eqnarray}
Here estimate \eqref{eq13} was used. The operator $(I + T)^{-1}$ is bounded by the Fredholm alternative because the equation $h + Th = 0$ has only the trivial solution for any $x \geq 0$. Indeed, by definition $A(x,y) = 0$ if $y < x$, so $A(x, x + t - p) = 0$ if $t < p$.
Thus, the equation $h + Th = 0$ can be considered as a convolution equation
\begin{equation}\label{eq23'}
h(p) + \int_{-\infty}^\infty A(x,x+t-p)h(t)dt = 0, \quad -\infty < p < \infty,
\end{equation}
where $h(t)=0$ for $t<2x$.
Taking the Fourier transform of equation  \eqref{eq23'}  and denoting $\tilde{h}(k) := \int_{-\infty}^\infty h(t)e^{itk}dt$ in the sense of distributions, one gets
\begin{equation}\label{eq24}
\tilde{h}(k) + \int_{-\infty}^\infty A(x,x+u)e^{-iuk}du  \tilde{h}(k) = 0, \quad \forall k \in \mathbb{R}.
\end{equation}
Since $A(x,x+u) = 0$ if $u < 0$, one can rewrite \eqref{eq24} as
\begin{equation}\label{eq25}
\left( 1 + \int_0^\infty A(x,x+u)e^{-iuk}du \right)\tilde{h}(k) = 0.
\end{equation}
By formula \eqref{eq2'}, the function $1 + \int_0^\infty A(x,x+u)e^{-iuk}du$ does not vanish for all real $k$ except, possibly, at $k = 0$, since the function $f(x,k)$ does not vanish for any fixed $x \geq 0$ as a function of $k$ on the sets of positive Lebesgue's measure on $\mathbb{R}$.
Note that the function $\int_0^\infty A(x,x+u)e^{-iuk}du$ belongs to the Hardy class in the half-plane Im$k \leq 0$ (see, for example, \cite{K}), because $A(x,x+u) \in L^1(\mathbb{R}_+)$ if estimate \eqref{eq14} holds. Since $f(k)$ does not vanish on sets
of positive Lebesgue's measure on $\mathbb{R}$, it follows that
  equation \eqref{eq25} implies $\tilde{h} = 0$, so $h = 0$, and the operator $(I + T)^{-1}$ is boundedly invertible
in $L^p$, $p=1,2,$ that is, $||(I + T)^{-1}|| \leq c_0$ for all $x \geq 0$. Consequently, equation  \eqref{eq21} implies:
\begin{equation}\label{eq26}
\int_0^\infty |F(p)|dp \leq c_0\int_0^\infty |A(0,p)|dp \leq c_0c\int_0^\infty dp \int_{p/2}^\infty |q(s)|ds < \infty.
\end{equation}
Since $F=F_s+F_d$ and $F_d(x) \in L^1$, it follows from \eqref{eq26} that
\begin{equation}\label{eq27}
\int_0^\infty |F_s(p)|dp < \infty.
\end{equation}
One can also prove that $\displaystyle\int_{-\infty}^0 |F_s(p)|dp < \infty$, but we do not use this estimate.\\

{\em Let us prove that $xF'(x) \in L^1(\mathbb{R}_+)$.} Write the basic equation \eqref{eq5} as
\begin{equation}\label{eq28}
A(x,x+t) + F(2x+t) + \int_0^\infty A(x,x+p)F(2x+p+t)dp = 0, \quad t \geq 0, x \geq 0.
\end{equation}
Since the operator $(I + Q)^{-1}$ is bounded in $L^1(\mathbb{R}_+)$ for all $x \geq 0$, where
$$Qh := \displaystyle\int_0^\infty F(2x+p+t)h(p)dp,$$
 it follows that
\begin{equation}\label{eq29}
\int_0^\infty |A(x,x+t)|dt \leq c \int_0^\infty |F(2x+t)|dt < \infty.
\end{equation}
Denote  $\dot{A}(x,x) := \displaystyle\frac{dA(x,x)}{dx}$. Let $y = x$ in equation \eqref{eq5} and differentiate
this equation with respect to $x$ to get:
\begin{multline}\label{eq30}
\dot{A}(x,x) + 2F'(2x) - A(x,x)F(2x) + \int_x^\infty A_x(x,s)F(s + x)ds\\ +
 \int_x^\infty A(x,s)F'(s+x)ds = 0.
\end{multline}
Integrate by parts the last integral in \eqref{eq30} to get
\begin{equation}\label{eq31}
\dot{A}(x,x) + 2F'(2x) - 2A(x,x)F(2x) + \int_x^\infty \left[ A_x(x,s) - A_s(x,s) \right] F(s + x)ds = 0.
\end{equation}
One gets from \eqref{eq31} the desired inclusion  $xF'(x) \in L^1 := L^1(\mathbb{R}_+)$  provided that:

 a) $x\dot{A}(x,x) \in L^1$,

 b) $xA(x,x)F(2x) \in L^1$,

 and

 c) $x\int_x^\infty \left[ A_x(x,s) - A_s(x,s) \right] F(s + x)ds \in L^1$.\\
The inclusion a) follows from formula \eqref{eq6} and the assumption $q \in L_{1,1}$. The inclusion b) follows from the estimate
\begin{equation}\label{eq32}
|xA(x,x)| = \left| -\frac{1}{2}x\int_x^\infty q(s)ds \right| \leq \frac{1}{2}\int_x^\infty s |q(s)|ds < \infty, \quad \forall x \geq 0,
\end{equation}
and from the inclusion $F \in L^1$.\\
Let us check the inclusion c). From the equation \eqref{eq6'}, it follows that
\begin{eqnarray}
A_x(x,y) = -\frac{1}{4}q\left(\frac{x+y}{2}\right) + \frac{1}{2}\int_x^\infty ds q(s) [-A(s,y+s-x) - A(s,y-s+x)], \label{eq34} \\
A_y(x,y) = -\frac{1}{4}q\left(\frac{x+y}{2}\right) + \frac{1}{2}\int_x^\infty ds q(s) [A(s,y+s-x) - A(s,y-s+x)]. \label{eq35}
\end{eqnarray}
Thus,
\begin{equation}\label{eq36}
A_x(x,s) - A_s(x,s) = -\int_x^\infty dt q(t) A(t,s+t-x).
\end{equation}
From \eqref{eq36} and \eqref{eq13} one gets
\begin{multline}\label{eq37}
|A_x(x,s) - A_s(x,s)| \leq c\int_x^\infty dt |q(t)|\int_{t + \frac{s - x}{2}}^\infty |q(p)|dp \\
\leq c\left( \int_x^\infty |q(t)|dt \right)^2 := c\tau^2(x).
\end{multline}
Therefore, inclusion c) holds if
\begin{equation}\label{eq38}
x\int_x^\infty \tau^2(x)F(s + x)ds \in L^1.
\end{equation}
Denoting by $c>0$ various estimation constants, one gets:
\begin{eqnarray}
\max_{x \geq 0}x\tau(x) \leq \int_x^\infty t|q(t)|dt \leq \int_0^\infty t|q(t)|dt \leq c; \quad \int_0^\infty \tau(x)dx \leq c; \label{eq39} \\
\int_0^\infty dx x\tau^2(x)\int_x^\infty |F(s + x)|ds \leq c \int_0^\infty |F(s)|ds < \infty.
\end{eqnarray}
Therefore, conditions \eqref{eq10} hold if $q \in L_{1,1}$. The necessity  of our conditions \eqref{eq8}-\eqref{eq11} is proved. \hfill$\Box$\\

\noindent {\em Sufficiency.} Assume now that conditions \eqref{eq8} - \eqref{eq11} hold and let us prove that
the corresponding scattering data come from a potential $q \in L_{1,1}$.

From conditions \eqref{eq8} and \eqref{eq11} it follows that the Riemann problem
\begin{equation}\label{eq41}
f(k) = S(-k)f(-k)
\end{equation}
is solvable. Here $f(k)$ is analytic in $\mathbb{C}_+$, $f(\infty) = 1$, $f(-k)$ is analytic in
$\mathbb{C}_-$, $f(-k) = \overline{f(k)}$ when $k \in \mathbb{R}$.   The solution to this Riemann problem is unique if
\begin{equation}\label{eq42}
f(ik_j) = 0, \quad \dot{f}(ik_j) \neq 0, \quad 1 \leq j \leq J, \quad \text{ind}S(k) = -2J.
\end{equation}
If ind$S(k) = -2J - 1$ then we require in addition to \eqref{eq42} that
\begin{equation}\label{eq43}
f(0) = 0, \quad \dot{f}(0) \neq 0.
\end{equation}
To prove the existence and uniqueness of the solution $f(k)$ to the Riemann problem \eqref{eq41}, satisfying conditions \eqref{eq42}
(or conditions \eqref{eq42}-\eqref{eq43} if ind$S(k)=-2J-1$, in other words, if $f(0)=0$),  let us introduce the functions
\begin{equation}\label{eq44}
w(k) := \prod_{j = 1}^J\frac{k - ik_j}{k + ik_j}; \quad w_0(k) := w(k)\frac{k}{k + i\kappa},
\end{equation}
where $\kappa>0$ is a number, $\kappa \neq k_j, \, \forall j = 1, \dots, J$.
If $f(0) \neq 0$, then equation \eqref{eq41} can be rewritten as
\begin{equation}\label{eq45}
\phi_+(k) := \frac{f(k)}{w(k)} = \frac{S(-k)w(-k)}{w(k)} \quad \frac{f(-k)}{w(-k)} := \frac{S(-k)}{w^2(k)}\phi_-(k),
\end{equation}
where we took into consideration that
\begin{equation}\label{eq46}
w(-k) = \frac{1}{w(k)}, \quad w(-k) = \overline{w(k)}, \quad k \in \mathbb{R}.
\end{equation}
If $S(k)$ satisfies \eqref{eq8} then $\frac{S(-k)}{w^2(k)}$ satisfies \eqref{eq8} and
 ind$\frac{S(-k)}{w^2(k)} = 0$ if ind$S(k) = -2J$. Indeed, ind$S(-k) = -$ind$S(k) = 2J$ and
 ind$w^2(k) = 2$ind$w(k) = 2J$. Since ind$\frac{S(-k)}{w^2(k)} = 0$ and $S(k) \neq 0$, equation \eqref{eq45}
implies
\begin{equation}\label{eq47}
\ln\phi_+ = \ln\frac{S(-k)}{w^2(k)} + \ln\phi_-(k),
\end{equation}
where $\ln\phi_+$ is analytic in $\mathbb{C}_+$ and $\ln\phi_-$ is analytic in $\mathbb{C}_-$.
Consequently,
\begin{eqnarray}
\phi_+(k) = e^{ \frac{1}{2\pi i}\int_{-\infty}^\infty \frac{\ln S(-t)}{w^2(t)} \frac{dt}{t - k} }, \quad \text{Im}k > 0, \label{eq48} \\
\phi_+(k + i0) = e^{ \frac{1}{2\pi i}\text{v.p.}\int_{-\infty}^\infty \frac{\ln S(-t)}{w^2(t)} \frac{dt}{t - k} +
\frac{1}{2}\ln\frac{S(-k)}{w^2(k)} }, \quad k \in \mathbb{R}, \label{eq48}
\end{eqnarray}
and
\begin{equation}\label{eq49}
f(k) = w(k)\phi_+(k),
\end{equation}
with
\begin{equation}\label{eq50}
f(-k) = \overline{f(k)}, \quad k \in \mathbb{R}.
\end{equation}
Consider now the case ind$S(k) = -2J - 1$, and look for $f$ such that $f(0) = 0, \dot{f}(0) \neq 0$.
In this case the argument is similar to the one used in the case $f(0) \neq 0$, but the function $w(k)$ is replaced by $w_0(k)$, and one has:
\begin{equation}\label{eq51}
\text{ind}w_0(k) = 2J + 1, \quad \text{ind}\frac{S(-k)}{w_0(k)} = 0.
\end{equation}
Denote  $\dot{f} = \frac{df(k)}{dk}$.
Let us summarize what we have proved.
\begin{lemma}\label{Lem1}
If conditions \eqref{eq8} and \eqref{eq11} hold, then $S(k) = \frac{f(-k)}{f(k)}$, where $f(k)$ is analytic
in $\mathbb{C}_+$, $f(\infty) = 1$, $f(ik_j) = 0$, $\dot{f}(ik_j) \neq 0$, $1 \leq j \leq J$.
 If ind$S(k) = -2J-1$, then $f(0) = 0$, $\dot{f}(0) \neq 0$.
\end{lemma}

 Let us prove that the potential \eqref{eq6}, obtained by solving equation \eqref{eq5}, belongs to $L_{1,1}$.
  The function $F(x)$, defined by formula \eqref{eq4}, is real-valued because of the assumptions \eqref{eq8} and \eqref{eq9}.
   Therefore, the solution $A(x,y)$ to equation \eqref{eq5} is real-valued and so is potential \eqref{eq6}.\\
The integral operator with the kernel $F(s+y)$ in \eqref{eq5} is compact in $L^1$ due to the first assumption \eqref{eq10}.
Thus, by the Fredholm alternative, equation \eqref{eq5} has a unique solution if the equation
\begin{equation}\label{eq52}
h(y) + \int_x^\infty F(s+y)h(s)ds = 0, \quad \forall x \geq 0,
\end{equation}
has only the trivial solution. Here $F(x)$ is defined in \eqref{eq4}. A proof that equation \eqref{eq52} has only the trivial
 solution in $L^1$ is given in \cite{R470} (and in \cite{M}, by a different argument). Therefore, equation \eqref{eq5} defines
 uniquely the kernel $A(x,y)$, and, consequently,
 the potential $q(x)$ is uniquely determined by formula \eqref{eq6}. Let us prove the inclusion $q \in L_{1,1}$.
 This inclusion follows from the assumptions \eqref{eq10} and equation \eqref{eq5}
as we prove below.
Let us rewrite equation \eqref{eq5} in the form \eqref{eq28} and differentiate equation \eqref{eq28}  with respect to $x$,
 denoting $\frac{dA(x,x+t)}{dx} := \dot{A}(x,x+t)$. The result is:
\begin{multline}\label{eq53}
\dot{A}(x,x+t) + 2F'(2x+t) + \int_0^\infty \dot{A}(x, x+p)F(2x + t + p)dp \, + \\ + 2\int_0^\infty A(x,x+p)F'(2x+t+p)dp = 0.
\end{multline}
From the bounded invertibility of the operator $I + Q$ in equation \eqref{eq28}, where $Q$ is compact in $L^1$,
\begin{equation}\label{eq54}
Qh := \int_0^\infty F(2x + t + p)h(p)dp,
\end{equation}
it follows from equation \eqref{eq28} that estimate \eqref{eq29} holds for all $x\ge 0$.\\
Let us define
\begin{equation}\label{eq55}
w(x) := \int_x^\infty |F'(t)|dt, \quad w_1(x): = \int_x^\infty w(t)dt.
\end{equation}
If the second assumption \eqref{eq10} holds, then $\max_{x \geq 0}w_1(x) \leq c$. Note that
\begin{multline}\label{eq56}
|F(t)| \leq \int_t^\infty |F'(s)|ds = w(t); \quad xw(x)\le \int_x^\infty t|F'(t)|dt<c; \\
\int_0^\infty w(x)dx\le \int_0^\infty t|F'(t)|dt<c.
\end{multline}
Equation \eqref{eq28} implies
\begin{equation}\label{eq57}
|A(x,x+t)| \leq w(2x+t) + c_1w(2x+t) \leq cw(2x + t).
\end{equation}
By $c > 0$, $c_1>0$ and $c_2>0$ various estimation constants are denoted. From equation \eqref{eq53} and the
 boundedness of the operator $(I + Q)^{-1}$ in $L^1$, one gets
\begin{multline}\label{eq58}
\int_0^\infty |\dot{A}(x, x+t)|dt \leq c\left( \int_0^\infty |F'(2x+t)|dt +
 \int_0^\infty dt\int_0^\infty |A(x, x+p)|\cdot \right. \\ \cdot |F'(2x+t+p)|dp \bigg) \leq c\Big( w(2x) + w(2x)c \Big) \leq c_2w(2x).
\end{multline}
Equation \eqref{eq53} with $t = 0$ implies $x\dot{A} \in L^1$ provided that $xF' \in L^1$ and
\begin{equation}\label{eq59}
J_1 := x\int_0^\infty |\dot{A}(x, x+p)| \, |F(2x+p)|dp \in L^1,
\end{equation}
and
\begin{equation}\label{eq60}
J_2 := x\int_0^\infty |A(x,x+p)| \, |F'(2x+p)|dp \in L^1.
\end{equation}
The inclusion  $xF' \in L^1$ is part of the condition \eqref{eq10}.
Let us prove \eqref{eq59} and \eqref{eq60}. Using estimates \eqref{eq56} - \eqref{eq58}, one gets
\begin{equation}\label{eq61}
\int_0^\infty dx x \int_0^\infty |\dot{A}(x, x+p)|dp\, w(2x) = \int_0^\infty dx x w^2(2x) \leq c,
\end{equation}
where we have used the estimates
\begin{equation}\label{eq62}
\max_{x \geq 0}xw(x) \leq c; \quad \int_0^\infty dx w(x) \leq \int_0^\infty dx \int_x^\infty |F'(t)|dt = \int_0^\infty t|F'(t)|dt < c.
\end{equation}
Thus, relation \eqref{eq59} is proved. Let us prove \eqref{eq60}. One has
\begin{equation}\label{eq63}
\int_0^\infty dx x \int_0^\infty |A(x,x+p)| \, |F'(2x+p)|dp \leq \int_0^\infty dx x w(2x)w(2x) < c.
\end{equation}
Thus, relation \eqref{eq60} is proved and the relation $q \in L_{1,1}$ is established.
Theorem 1.2 is proved.
\end{proof}

\end{document}